\documentclass{article}

\usepackage{setspace}

\usepackage[a4paper, portrait, margin=1.1811in]{geometry}
\usepackage[utf8]{inputenc}
\usepackage{amsmath,amsfonts, amssymb,amsthm}
\usepackage{newtxmath}
\usepackage{wrapfig,xcolor,graphicx,caption}
\usepackage{diagbox,wasysym}
\usepackage{multicol,enumitem,booktabs}
\usepackage{authblk}
\usepackage{url}
\usepackage{dsfont}
\usepackage{mathtools}
\usepackage[authoryear]{natbib} 
\usepackage[colorlinks=true,citecolor=blue!50!black,linkcolor=blue!50!black,urlcolor=blue!50!black,pdfauthor=author]{hyperref}
\usepackage[capitalize,nameinlink,noabbrev]{cleveref}

\def\namedlabel#1#2{\begingroup
   \def\@currentlabel{#2}%
   \label{#1}\endgroup
}

\mathtoolsset{showonlyrefs,showmanualtags}

\DeclareMathOperator*{\argmin}{argmin}
\DeclareMathOperator*{\argmax}{argmax}
\DeclareMathOperator*{\argsup}{argsup}

\newcommand{\R}{\mathbb{R}}
\newcommand{\1}{\mathds{1}}
\newcommand{\C}{\mathcal{C}}
\newcommand{\E}{\mathcal{E}}
\renewcommand{\P}{\mathcal{P}}
\renewcommand{\S}{\mathfrak{S}}
\renewcommand{\dim}{m}

\theoremstyle{plain}

\newtheorem{proposition}{Proposition}
\newtheorem{definition}{Definition}

\theoremstyle{remark}
\newtheorem{remark}{Remark}

\title{Deepest voting on rankings}
\author[1]{Jean-Baptiste Aubin}
\author[2]{Antoine Rolland} 
\author[3]{Ioana Gavra}
\author[4]{Irène Gannaz}
\author[4]{Jacques Anderson Kouassi}

\affil[1]{Univ Lyon, INSA Lyon, UJM, UCBL, ECL, ICJ, UMR5208, \protect\\ 69621 Villeurbanne, France}
\affil[2]{Université Lumière Lyon 2, Universite Claude Bernard Lyon 1, ERIC, 69007, Lyon, France}
\affil[3]{Univ Rennes 2, IRMAR - UMR CNRS 6625, F-35000 Rennes, France,}
\affil[4]{ Univ. Grenoble Alpes, CNRS, Grenoble INP, G-SCOP, 38000 Grenoble, France}

\date{February 2026}

\begin{document}

\maketitle

\setlength{\parskip}{12pt} 
\setlength{\parindent}{0pt} 
\onehalfspacing

\section*{Abstract}
This article aims to present a unified framework for ranking-based voting rules based on the use of depth functions on permutations, as a counterpart of deepest voting rules on evaluation introduced in \cite{aubin2022deepest}. It introduces the notion of depth functions, in continuous sets  and in permutation sets, the later using the notion of Fréchet means. Deepest voting procedures are then formally  defined, and some classical voting rules are expressed as deepest voting procedures, using a large variety of distances on the set of permutations. Links are done between the depth functions mathematical properties and some behaviours of the voting rule, such as Neutrality, Anonymity,  Universality, Condorcet winner/loser property and so on.

\section{Introduction}

A voting procedure, or social choice function, is a process that aggregate individual judgments into collective outcomes. In case of uninominal voting, the result should be the election on an unique winner selected in a finite set of candidates. Traditionally, voting theory has been dominated by preference-based (ordinal) methods, where each voter is supposed to ranks candidates by preference, and the social choice function selects a winner or social ranking accordingly. Classic examples include plurality voting or Borda count, which interpret voter rankings to reflect collective preferences and satisfy normative criteria such as majority support or monotonicity. Social choice theory rigorously formalizes these models, exploring their properties and limitations — most famously exemplified by Arrow’s impossibility theorem and related strategic concerns such as the Gibbard–Satterthwaite result on manipulability of ordinal rules.

Another framework is although possible, which concerns evaluation-based (cardinal) voting procedures. Voters are supposed to assign scores or grades to each option, providing richer information about intensity of support than in the case of ranking. Cardinal methods, such as scores or range voting, are more convenient to  capture nuances in voter sentiment, which should improve collective welfare outcomes..

In the context of evaluation-based voting rules, \cite{aubin2022deepest} propose an unified model of social choice function based on the use of the depth function concept. A depth function is a measure of the "centrality" of a point into a scatter-plot. The deepest point  is therefore the most central point of the scatter-plot. Of course there are many different depth functions that lead to possibly different deepest points. For example in a one-dimension space, both the sum of absolute differences or the sum os square differences to the other points can be considered as depth functions (strictly speaking, inverse of depth functions) leading to the median and mean as deepest points.

So roughly speaking, \cite{aubin2022deepest} introduce deepest voting rules, that consider voters as points in the candidates' space, and tries to recover the "most central" voter in the voters set. This specific voter is therefore considered as most representative of the collective evaluations, and the choice of this vote is considered to be the collective choice. The choice of a specific depth function leads to the definition of specific voting rule, including well-known evaluation-based voting rules such as range voting or majority judgment.

The objective of this paper is to generalize the deepest voting approach to ranking-based votes.  \cite{aubin2022deepest} only consider evaluations based settings. We propose in this paper to formalize ranking-based voting rules through the use of depth functions defined on rankings seen as permutations. We use the concept of $p-$Fréchet mean to generalize the concept of depth functions to permutations. We show that some usual procedures (namely Bucklin, Borda, Kemeny, Plurality, Antiplurality) can be written as deepest voting processes. We also analyse some classical voting rules properties that are  satisfied or not by ranking-based deepest voting rules.

Distance rationalization of voting rules is also an axiomatic approach based on the use of distance between rankings to determine the result of a voting process, see for example \cite{Meskanen2008}, \cite{ElkindTARK} and \cite{Elkind2010OnTR}. However the framework is  different. Rationalization of voting rules is based on the idea to compare a global voting profile to the closest unambiguous voting profile (for example unanimous profiles in the early work of  \cite{Nitzan81}), that leads to an unambiguous winner of the election.  For example Dodgson's voting rule is the rationalization of the Condorcet winner with respect to the Kendall distance; how much should one derive of the observed election profile to reach a Condorcet winner? Deepest voting, on the contrary, only measure the inner distances between voters to find the innermost voter, which is the most representative of the election profile.

More recently the Level $r$ Consensus method proposed in  \cite{MahajneNV15}  deals with distance between preference relations w.r.t. an election profiles, and not between voters as in the deepest voting rules.

Formally, we consider a framework with $n$ voters, $\mathcal V = \lbrace 1,\dots n\rbrace$, and $\dim$ candidates $\mathcal C=\lbrace 1,\dots \dim \rbrace $. A voting situation suppose each voter $v$ gives an opinion on candidate $c$, denoted $e_{cv}$. The opinions $(e_{cv})_{c\in\mathcal C}$ take values in $\mathcal E$. In preferences ranking votes, the set $\mathcal E$ is the set of preferences. That is, $\E = \S_\dim$, with $\S_\dim$ the set of permutations of $\{1,\dots, \dim\}$. The idea of evaluation-based framework is to change the set of opinions to evaluations rather than rankings. In such a case, $\mathcal E = \Lambda^\dim$, with $\Lambda$ the set of evaluations of one voter for each candidate. Classical choices are $\Lambda=\lbrace 1,\dots, K \rbrace$ (discrete evaluations) and $\Lambda=[0,M]$ (continuous evaluations). In the following, we will consider that an evaluation decreases with the preference to a candidate. That is, if the voter $v$ prefers the candidate $c_1$ to the candidate $c_2$, then $e_{c_1 v}<e_{c_2 v}$. This ordering being the same as the one on rankings allows to make parallelisms. 

Hence, a voting framework consists in votes $\Phi=(e_{cv})_{v\in\mathcal{V},c\in\mathcal{C}}$. For each voter $v$, the opinion $(e_{cv})_{c\in\mathcal{C}}$ can be seen as the realisation of a $\dim$-multivariate random variable $E$, with $E$ at values in $\E$. This statistical point of view offers an interesting perspective. It allows using the statistical tools to describe the set of votes. In particular, it offers the possibility of finding \emph{centers} of the distribution, thanks to the application of statistical depth functions.

Quoting \cite{ZuoSerfling}: ``Associated with a given distribution $F$ on $\mathbb{R}^\dim$, a depth function is designed to provide a $F$-based center-outward ordering [...] of points $x$ in $\mathbb{R}^\dim$. High depth corresponds to \emph{centrality}, low depth to \textit{outlyingness}''. In other words, a depth function takes high (positive) values at the \emph{center} of a scatter plot and vanishes out of it.  A formal definition will be given in \cref{sec:depths}. A depth function $D$ on $\E$ is a function defined on $\E\times \mathcal{F}$, where $\mathcal{F}$ is the set of probability distributions on $\E$. It has values in $[0,\infty)$, and is built such that for any distribution $F$ on $\E$, the function $D(.,F)$ is maximal at a point which can be considered as the \emph{center} of the distribution $F$. Applied to an observed vote situation, $\Phi\in\E^n$, $D(.,\Phi)$ is maximal at a point which can be considered as the \emph{center} of the empirical distribution related to $\Phi$. Using this approach enables to recover classical \emph{centers} of distributions, such as the mean or the median, but also to define various \emph{centers}, having interesting properties.

 The application of a depth function $D()$ on a voting situation $\Phi\in\E^n$ provides a (possibly fictive) central voter, $v^*$ with opinions $E^*\in\E^\dim$, $E^*=(e^*_{1},\dots e^*_{\dim})^\top$. The preferred candidate of $v^*$ is then the winner of the votes, associated to the depth function. This procedure is defined as deepest voting procedure by  \cite{aubin2022deepest}. We investigate in this paper the case where $\E=\S_\dim$.

The paper is organized as follows. \cref{sec:depths} provides the definitions of depth functions, in continuous sets $\R^\dim$ and in permutation sets $\S_\dim$. The later are in particular based on the notion of Fréchet means and on distances on $\S_m$, that will be recalled in this section. Deepest voting procedures are then defined in \cref{sec:deepest_voting}. \cref{sec:eval} then gives some results with continuous depths, while \cref{sec:ranks} consider depth based on permutations. In particular, some classical voting rules are expressed as deepest voting procedures, and links are done between the depth functions mathematical properties and some behaviours of the voting rule, such as Neutrality, Anonymity or Universality. Proofs are given in the Appendix.

\section{Definition of depth functions}
\label{sec:depths}

In this section, we provide formal definitions of depths functions. We first consider depths functions defined on continuous multivariate sets, on a connected subset of $\R^\dim$. Next we present depth functions on the set of permutations $\S_n$.

\subsection{Depth functions in continuous case}
\label{sec:depths_continu}

We recall here the definition of a depth function on $\R^\dim$. Such a depth function can be applied in any connected set $\E_D\subseteq \R^\dim$. We refer to \cite{mosler2013depth, ZuoSerfling} and references therein for a more detailed overview.

\begin{definition} \label{def:depth_R} Let the mapping $D:\mathbb{R}^\dim \times \mathcal{F} \to \R^+$ be bounded, and satisfying:
\begin{enumerate}[label=(C\arabic*),topsep=0pt]
\item Stability by Permutation. Let $E=(e_1,\dots,e_\dim)$ be a random vector in $\R^\dim$, $x\in\R^\dim$, and $\sigma$ a permutation on $\{1,\dots,\dim\}$. Let $E_\sigma=(e_{\sigma(1)},\dots,e_{\sigma(\dim)})$ and $x_\sigma=(x_{\sigma(1)},\dots,x_{\sigma(\dim)})$. Then $D(x_\sigma,F_{E_\sigma})=D(x, F_E)$.
\label{ass:depth0}

\item Affine Invariance. For all $a\in\R$, $b\in\R^\dim$, for any random vector $E \in \mathbb{R}^\dim$, $\argmax_{x\in\R^\dim} D(a\, x+b,F_{a\,E+b})= \argmax_{x\in\R^\dim} D(x,F_{E})$, where $F_X$ denotes the distribution of the random variable $X$.
\label{ass:depth1}

\item Maximality at center. For a distribution $F \in \mathcal{F}$ having a uniquely defined \emph{center} $\theta$ (\emph{e.g.} the point of \emph{symmetry}), $D(\theta,F)=\sup_{x \in \mathbb{R}^\dim} D(x,F)$.
\label{ass:depth2} 

\item Quasi-concavity. For any $F \in \mathcal{F}$, $D(\cdot,F)$ is quasi-concave. That is, if $\theta\in\argmax_{x\in\R^\dim} D(x,F)$, then $D(x,F) \leq D(\theta + \lambda (x-\theta), F)$ for any $0 \leq \lambda \leq 1$. \label{ass:depth3}

\item Vanishing at Infinity. Let $\lVert .\rVert$ denote the euclidean norm on $\R^\dim$. Then, $D(x,F) \rightarrow 0$ as $ \|x\| \rightarrow \infty$ for each  $F \in \mathcal{F}$. \label{ass:depth4}

\item Let $F \in \mathcal{F}$ be a distribution on $\R^\dim$ with marginal distributions $F_1,\dots, F_\dim$. Suppose that for $c\in\{1,\dots,\dim\}$, $F_c$ has support containing a unique point $\{\alpha\}$. Then for all $x^*\in\argsup_{x \in \mathbb{R}^\dim} D(x,F)$, the $c^\text{th}$ coordinate of $x^*$ is $x^*_c=\alpha$. \label{ass:depth5}

\end{enumerate}
Then $D(\cdot,\cdot)$ is called a {\bf statistical depth function}.

\end{definition}

We refer to \cite{aubin2022deepest} for a discussion on the assumptions. Note that the notion of symmetry in \ref{ass:depth2} is (voluntarily) not precisely fixed, and various notions of symmetry are possible such as, from the most constraining to the weakest, central symmetry, angular symmetry and halfspace symmetry. 

In our context, we will apply the depth functions on empirical distribution. Consider a voting situation with $n$ voters and $\dim$ candidates, $\mathcal C=\lbrace 1,\dots \dim\rbrace$, with an opinion matrix $\Phi\in\E^n$. Let $D()$ be a depth function defined on $\E_D\times \mathcal{F}$. In the following, we will denote $D(.,\Phi)$ the depth function $D(x,F_\Phi)$ where $F_\Phi$ is the empirical distribution of the voter's opinions in $\E$, $\{\Phi(.,v)=(e_{cv})_{c\in\{1,\dots,\dim\}},\, v=1,\dots,n\}$.

Let $x\in \mathbb{R}^\dim$ and $\Phi(.,1),\ldots, \Phi(.,n)$ in $\mathbb{R}^\dim$. Examples of depth functions satisfying \cref{def:depth_R} are:
\begin{description}[topsep=0pt]
\item[The weighted $L^q$ depths.] \citep{Zuo2004, mosler2013depth} The weighted $L^q$ depth is defined by
$$ w L^q D(x,\Phi)=\frac{1}{1+\frac{1}{n} \sum_{v=1}^n \omega( \|\Phi(.,v)-x\|_q)},  $$
where $q>0$, $\omega$ is a non-decreasing and continuous function on $[0,\infty)$ with $\omega(\infty)=\infty$
 and $\| x-x' \|_q =\left(\sum_{c=1}^\dim |x_c-x'_c|^q\right)^{1/q}$. 
 
If $\omega: x \mapsto x^q$, then 
\begin{equation}
\label{eqn:wLp} L^q D(x,\Phi):= \frac{1}{1+\frac{1}{n} \sum_{v=1}^n   \sum_{c=1}^\dim |\Phi(c,v)-x_{c}|^q}
\end{equation}
will be called a $L^q$ depth.

If $q=\infty$, the definition can be extended to
$L^\infty D(x,\Phi):= 1/\bigl(1+\frac{1}{n} \sum_{v=1}^n   \max_{c=1,\dots \dim} |\Phi(c,v)-x_{c}|\bigr).$

\item[The halfspace depth.] \citep{Tukey}  The halfspace depth is defined by 
\begin{equation*} HD(x,\Phi):= \textrm{ minimum proportion of voters in a halfspace $H$ including } x.\end{equation*}

\item[The projection depth] \citep{Zuo2003} The projection depth is defined by 
$$ PD(x, \Phi):= \inf_{u\in\R^\dim,\,\lVert u\rVert =1} \frac{1}{1+\lvert u^\top x -\mu(F_u)\rvert/\sigma(F_u)},$$
where $\lVert\cdot\rVert$ denotes the euclidean norm,, $\mu(F)$ denotes a central statistic of a distribution $F$ and $\sigma(F)$ a dispersion statistic. $F_u$ is the empirical distribution of $u^\top \Phi$. Typically, $\mu(\cdot)$ is the median and $\sigma(\cdot)$ is the median absolute deviation.
\end{description}

Note that following $L^q$-depths, a depth function can be built from a distance $d()$ on $\R^\dim$, considering $D(x,\Phi_n):= 1/\bigl({1+d(x_c,\Phi(c,v))}\bigr).$

\subsection{Depth functions on permutations}
\label{sec:depths_permut}

Our objective is to extend the deepest voting framework, defined in evaluation-based framework by \cite{aubin2022deepest}, to voting methods based on preference rankings. For preference rankings, opinions belong to the set $\E=\S_m$. We could consider depth functions defined on $\E_D=[0,\dim]^\dim$, since $\E\subset\E_D$, but it seems appropriate to explore depth functions defined on $\E_D=\S_\dim$. To achieve this goal, we want to define depth functions on permutation sets.

Before defining depth function, let us recall some definitions of distances on $\S_\dim$.

\subsubsection{Distances on permutations}\label{sec:dist_perm}

Various distances on the set $\S_\dim$ have been defined in the literature (see \cite{dist.permut} and  \cite{deza2016encyclopedia} for complete reviews). Especially, for $\sigma, \tau\in\S_\dim$, examples of distances are
\begin{description}[topsep=0pt]
    \item[Kendall's distance]:
    \begin{equation}
    \label{eq:kendall}
        d_K(\sigma,\tau)=\sum_{c=1}^{\dim-1}\sum_{c'=c+1}^{\dim}\1\{(\sigma(c)-\sigma(c'))\,(\tau(c')-\tau(c))< 0\}.
    \end{equation}
    
    \item[Hamming distance]: $$d_H(\sigma,\tau)=\sum_{c=1}^\dim\1\{\sigma(c)\neq\tau(c)\}.$$
    
    \item[Cayley metric] (or transposition distance): $d_C(\sigma,\tau)$ is the minimum number of transpositions needed to obtain $\sigma$ from $\tau$ (cf \cite{Diaconis88}),
$$d_C(\sigma, \tau)=\dim-\#\mbox{cycles in}(\sigma \circ \tau^{-1}),$$ where $\sigma\circ s$ denotes the successive applications of $s$ and $\sigma$: for all $c\in \{1,\dots,\dim\}$, $\sigma\circ s(c)=\sigma(s(c))$.

    \item[Minkowski-Hölder distances]:  
    \begin{align} \label{eq:minkowski}
\forall q\in[1,+\infty),      d_{q}(\sigma,\tau) &=\Bigl(\sum_{c=1}^\dim|\sigma(c)-\tau(c)|^{q}\Bigr)^{\frac{1}{q}},
     \\ 
     d_{\infty}(\sigma,\tau)& =\underset{c \in \{1, \ldots, \dim\}}{\max}\bigl\lvert\sigma(c)-\tau(c)\bigr\rvert.
  \end{align}

Among Minkowski-Hölder distances, two special cases are when $q=1$ and $q=2$:
 \begin{description}
    \item[when $q=1$], Minkowski-Hölder distance is referred to as Spearman footrule,
     \item[when $q=2$], Minkowski-Hölder distance is referred to as Spearman $\rho$ distance.%
     \end{description}
      \end{description}

\subsubsection{Weighted distances on permutations}
\label{sec:weights}

When measuring the discrepancy between two ranking lists in the context of voting, one might want to take some additional information into account: for example the positions on which the permutations differ or, if available, a similarity measure between the candidates. Some of the classical distances presented in \cref{sec:dist_perm} can be slightly modified by incorporating weights that balance the information in the distance. For instance, one could want to give more weight to the first ranked candidate with respect to the others, estimating the first position is more emblematic. Or one could want to build a distance where the weight of the difference between two given ranks increases with the difference. 
The idea is to introduce weights in the distance which will depends on the two ranks which are compared \cite{weighted_distances}. The weights write as $\lbrace w_{r,r'}, \; r,r'\in\lbrace 1,\dots, \dim\rbrace\rbrace$.

Applying this to Hamming distance and Minkowski-Holder type distance, we obtain, for $\sigma,\tau\in\S_\dim$,
\begin{itemize}[topsep=0pt]
    \item Weighted Hamming distance: 
\begin{equation}\label{eq:hamming_weighted}
d_{wH}(\sigma,\tau)=\sum_{c=1}^\dim w_{(\sigma(c),\tau(c))} \1\{\sigma(c)\neq\tau(c)\},
\end{equation}
    \item Weighted Minkowski-Holder type distance, for $q\in [1,+\infty)$:
   \begin{equation}\label{eq:minkowski_weighted}
   d_{wq}(\sigma,\tau)=\Bigl(\sum_{c=1}^\dim w_{(\sigma(c),\tau(c))}(|\sigma(c)-\tau(c)|)^{q}\Bigr)^{\frac{1}{q}}.\end{equation}
\end{itemize}
The weight $w_{r,r'}>0$ for each positions $r, r' \in \{1,\ldots,\dim\}$ balance the importance of swap, that is of moving from position $r$ to position $r'$ in the ranking.

While $d_{w_rH}()$ is indeed a distance on $\S_\dim$, in all generality, this is no longer the case for the functions $d_{wH}()$ and $d_{wq}()$ introduced above. In particular the triangle inequality does not necessarily hold. If the weights are symmetric, in the sense that $w_{r,r'}=w_{r',r}$, for all  $r, r' \in \{1,\ldots,\dim\}$, then the two dissimilarity functions are semi-metrics (they satisfy all the axioms of a distance except the triangle inequality). However, a so called $\rho$-relaxed triangle inequality does still hold and we have
$$d_{w}(\sigma_1,\sigma_2)\le \rho(d_{w}(\sigma_1,\sigma_3)+d_{w}(\sigma_3,\sigma_2)),$$
with $\rho=\max_{r_a,r_b,r_c}\frac{w_{r_a,r_b}}{w_{r_a,r_c}}$, and $d_{w}()$ equals $d_{wH}()$ or $d_{wp}()$.

\subsubsection{Depth functions associated to a distance}

Let $\P_\dim$ denote the set of probability distributions on $\S_\dim$.
Given a distance on $\S_\dim$, \cite{Goibert}  proposed the following definition of a depth function on permutations. 

\begin{definition}[\cite{Goibert}]
\label{def:Goibert_depths}
    Let $d()$ be a distance on the set $\S_\dim$ of permutations. 
    $D: \S_\dim\times \P_\dim \mapsto\R^+$ is a depth function if it satisfies the following properties:
    \begin{enumerate}[label={(P\arabic*)},topsep=0pt]
    \item Invariance. For all $s\in\S_\dim$ and $\Pi\in\P_\dim$, define $\Pi_{s}$ the probability distribution such that $\Pi_{s}(\sigma):=\Pi(\sigma\circ s^{-1})$ for all $\sigma\in\S_\dim$. 
    $D$ is said to be invariant if and only if, for all $\sigma, s\in\S_\dim$, for all $\Pi\in\P_\dim$,  $D(\sigma,\Pi)=D(\sigma\circ s,\Pi_s)$. 
    \label{ass:invariance}
    
    \item Maximality at Center. For a distribution $\Pi$ with a center of symmetry (defined later), the function $D(.,\Pi)$ reaches its maximum at this center. 
    \label{ass:maximality}
    
    \item Local Monotonicity. Let $\Pi\in\P_\dim$ and define $\sigma^*=\argmax_{\sigma\in\S_\dim} D(\sigma,\Pi)$. Suppose that the deepest permutation $\sigma^*$ is unique. Then for any $\tau, \sigma \in\S_\dim$ such that $d(\sigma^*,\sigma \circ \tau)=d(\sigma^*,\sigma)+1$, we have $D(\sigma\circ \tau,\Pi)\leq D(\sigma,\Pi)$.
    \label{ass:local_mono}
    
    \item Global Monotonicity. Let $\Pi\in\P_\dim$ and define $\sigma^*=\argmax_{\sigma\in\S_\dim} D(\sigma,\Pi)$. Suppose that the deepest permutation $\sigma^*$ is unique. Then we have $d(\sigma^*,\sigma)\leq d(\sigma^*,\tau) \ \Rightarrow \ D(\tau,\Pi) \leq D(\sigma,\Pi)$.
    \label{ass:global_mono}
    
    \end{enumerate}
\end{definition}

The properties are quite similar to \cref{def:depth_R}. First, \ref{ass:depth1} does not apply in the case of permutations. Indeed, affine transforms are undefined on the set $\S_m$. Also, since the set is countable, \ref{ass:depth4} does not apply either.
Next \ref{ass:invariance} is similar to \ref{ass:depth0}, and \ref{ass:local_mono} and \ref{ass:global_mono} correspond to \ref{ass:depth3}. 

Finally \ref{ass:maximality}, as  \ref{ass:depth2} in the continuous case, states that the depth function is maximal at a \emph{natural} center if such a center exists. \cite{Goibert} define a H-center (in Proposition 11), which corresponds to a halfspace symmetry. The authors also define a M-center (in Definition 6), based on the distance $d()$ on $\S_m$. We refer to the later for precise definitions and discussions.

These properties allow the depth function to be indeed maximal at the \emph{center} of a sample of permutations, and to behave similarly to a depth function defined on $\R^\dim$.

Depth functions in \cref{def:Goibert_depths} are highly related to Fréchet means.
\cite{Frechet} introduces the notion of typical position of order $p$, $p\geq 1$,
for distributions on general metric spaces as an extension of the classical moments of order $p$ associated with distributions on Euclidean spaces. This is known as $p$-Fréchet mean. In the case of permutations, the function associated to the $p$-Fréchet mean can be seen as a depth function on the metric space  $(\mathfrak{S}_m,d)$ as shown below.

\begin{definition}\label{def:Frechet_mean}
    Let $(\S_\dim,d)$ be a metric space endowed with a probability measure $\Pi$. For $p\geq 1$, the $p$-Fréchet mean on $(\S_\dim,d)$ with respect to $\Pi$ is defined as
    \begin{equation*}
        \argmin_{\sigma \in \S_\dim} U_{d,\Pi,p}(\sigma), \mbox{  with  } U_{d,\Pi,p}(\sigma) = \mathbb{E}_{S\sim\Pi}\left[ d^p(\sigma,S)\right].
    \end{equation*}
\end{definition}

Based on this notion, \cite{Goibert} defined some ranking depth functions, associated to metrics on $\S_\dim$.

\begin{definition}[\cite{Goibert}]
\label{def:depth_Goibert}
 Let $\S_{\dim}$ be the set of permutations on $\{1,\dots,\dim\}$. Let $d()$ be a distance on $\S_{\dim}$ and $p\in\R$, $p\geq 1$ a given parameter. The depth function associated to $(d,p)$ is defined as follows: $\forall \sigma \in \S_{\dim}$, for all $\Pi\in\P_\dim$,
 $$D(\sigma,\Pi)=\lVert d^p\rVert_{\infty}-\mathbb{E}_{S\sim\Pi}[d^p(\sigma,S)]=\lVert d^p\rVert_{\infty}- U_{d, \Pi,p}(\sigma),$$ 
 with $\lVert d^p\rVert_{\infty}=\max_{\sigma, \tau \in \S_\dim}d^p(\sigma, \tau)$, $S$ a random variable of law $\Pi$, and $U()$ is defined in \cref{def:Frechet_mean}.
\end{definition}

 The authors showed that, under sufficient conditions on the distribution $\Pi$, the depth function given in the definition above associated to the Kendall distance (see \eqref{eq:kendall}) satisfies the properties of \cref{def:Goibert_depths} for any $p\geq 1$. They also proved that \ref{ass:invariance} and \ref{ass:maximality}  hold when using Spearman footrule or Spearman $\rho$-distance (\eqref{eq:minkowski} with respectively $q=1$ and $q=2$). 
 
The deepest point can then be defined as the permutation maximizing the depth function. \cite{Goibert} introduce the \emph{consensus ranking} and the \emph{barycenter ranking}, which corresponds to the deepest points obtained respectively with $p=1$ and with $p=2$ in \cref{def:depth_Goibert}. We will focus on these two cases in the following. They are related to the $p$-Fréchet means as follows.

\begin{definition}[\cite{Goibert}]
\label{def:consensus_barycenter}
Let $(\S_\dim,d)$ be a metric space endowed with a probability measure $\Pi$. Denote $D_1()$ and $D_2()$ the depth functions associated respectively to $(d,1)$ and to $(d,2)$ as defined in \cref{def:depth_Goibert}. Then,
\begin{enumerate}[label=(\roman*),start=1,topsep=0pt]
\item a \emph{consensus ranking} $\sigma^* \in \S_\dim$ is a permutation such that $$D_1(\sigma^*,\Pi)=\max_{\sigma \in \S_\dim} D_1(\sigma,\Pi) = \lVert d\rVert_{\infty}-\min_{\sigma \in \S_\dim} U_{d, \Pi,1}(\sigma);$$
\item a \emph{barycenter ranking} $\sigma^* \in \S_\dim$ is a permutation such that
 $$D_2(\sigma^*,\Pi)=\max_{\sigma \in \S_\dim} D_2(\sigma,\Pi) = \lVert d^2\rVert_{\infty}-\min_{\sigma \in \S_\dim} U_{d, \Pi,2}(\sigma).$$
\end{enumerate}
\end{definition}

\begin{remark}
Note that $\mathbb{E}_\Pi[d(\sigma,\mathcal{S})]= U_{d,\Pi,1}(\sigma)$ means that, for any distance $d()$ and any probability distribution $\Pi$ on $\S_\dim$, the ranking given by maximizing the depth function $D_1(\sigma,\Pi)$ can be retrieved by minimizing the functional associated to the $1$-Fréchet mean on the same space. 
In particular, the consensus ranking given by the depth function $D_1(\sigma,\Pi)$ is a 1-Fréchet mean on $(\S_\dim,d,\Pi)$ and can be seen as a median ranking on $\S_\dim$. 
\end{remark}

\begin{remark}
Using the terminology introduced by \cite{ZuoSerfling},  $U_{d,\Pi,p}$ corresponds to the inverse of a \textit{Type B depth function}.
\end{remark}

Similarly to the continuous case, we will apply the depth functions on the empirical distribution of voting opinions. Consider a voting situation with $n$ voters and $\dim$ candidates, with an opinions matrix $\Phi\in\E^n$. If the opinions are rankings, $\Phi$ writes as $\Phi=(\sigma_v)_{v=1,\dots,n}$, and $\E=\S_\dim$. Let $D()$ be a depth function defined on $\E_D\times \P_\dim$. In the following, we will denote $D(.,\Phi)$ the depth function $D(x,F_\Phi)$ where $F_\Phi$ is the empirical distribution of the voter's opinions in $\E=\S_\dim$, $\{\Phi(.,v)=(e_{cv})_{c\in\{1,\dots,\dim\}},\, v=1,\dots,n\}$.

\section{Deepest voting}
\label{sec:deepest_voting}

The depth functions being defined, deepest voting procedures can be built as follows.

\begin{definition}[Deepest Voting] 
\label{def:deepest_voting}
Consider a voting situation with $n$ voters and $\dim$ candidates, $\C=\lbrace 1,\dots \dim\rbrace$, with an opinion matrix $\Phi\in\E^n$. Let $D(.,.)$ be a depth function defined on $\E_D\times \mathcal{F}$, where $\E\subseteq\E_D\subseteq \R^\dim$ and $\mathcal{F}$ is the set of probability distributions on $\E_D$. Denote 
\[\E^*_D:= \{E\in\R^\dim: D(E,\Phi)=\sup(D(.,\Phi)) \}\]
the set of deepest points (opinions) with respect to $\Phi$.\\
Let \[c_D^*:=\argmin_{c \in\mathcal{C}}  \{e^*_{D,c},~E^*_D=(e^*_{D,1},\hdots,e^*_{D,\dim})^\top\in\E_D^\ast\}.
\]
The deepest voting process with respect to the depth $D()$ is defined as the function which maps $\{\Phi(c,v),\; c=1,\dots,\dim,\; v=1,\dots,n\}$ to $c_D^*\subseteq \mathcal{C}$.\\
If $c_D^*$ is unique, then the winner of the election is the candidate $c_D^*$. If $c^*_D$ is not unique, there is no unique winner of the election.
\end{definition}

\cite{aubin2022deepest} consider the framework of evaluations. In such a case, the set of opinions $\E$ satisfies $\E=\Lambda^\dim$, with $\Lambda=\lbrace 0, \dots, K\rbrace$ or $\Lambda=[0,M]$. In that cases, the authors considered depth functions defined on $\E_D=\R^\dim$. 
Statistical depths functions are well defined in such spaces, which correspond to $\dim$-multivariate continuous distributions, as recalled in \cref{sec:depths_continu}. Some results in this framework are given in \cref{sec:eval}. The relation between usual voting procedures and given depth functions will be displayed. Note that in the discrete case depth functions defined on $\E_D=\mathbb{N}^\dim$ could have been considered, but, to our knowledge, such depth functions are not defined in literature.

The most common voting methods are based on preference rankings. This includes majority voting systems as well as those based on the Condorcet or Borda principles, among others. A review of these voting methods can be found in \cite{Felsenthal}. In that cases, the set of opinions $\E$ satisfies $\E=\S_\dim$, with $\S_\dim$ the set of permutations of $\{1,\dots, \dim\}$. The first possibility is to consider depth functions defined on $\E_D=[0,\dim]^\dim$, which indeed includes $\E=\S_\dim$. Usual depth functions used in the evaluation-based framework can then be used. An alternative is to consider depth functions defined on $\E_D=\E=\S_\dim$, as defined in \cref{sec:depths_permut}. This choice respects more the nature of the data. \cref{sec:ranks} show that some usual procedures can be written as permutation deepest voting. Relations between voting procedures and permutation depth functions are studied in \cref{sec:properties}.

\section{Continuous deepest voting}
\label{sec:eval}

Let us consider first evaluation-based voting. Recall that in such a case, the voters give opinions on candidates on the form of evaluation, that is, for each voter $v$ and candidate $c$, the opinion is $e_{cv}\in\Lambda$, with $\Lambda=\{1,\dots,K\}$ or $\Lambda=[0,M]$. Hence, each voter has a vector of opinions in $\E=\Lambda^\dim$. With $n$ voters and $\dim$ candidates, the opinion matrix, denoted $\Phi$, belongs to $\E^n$. In that case, we consider depth functions defined on $\R^\dim$, with properties given in \cref{def:depth_R}. For a given depth function $D()$, the associated deepest voting rule is obtained by maximizing $D(.,\Phi)$, as described in \cref{def:deepest_voting}.

In this context, \cite{aubin2022deepest} established parallelism between conditions on the depth functions and usual axioms on voting procedures. The authors prove that a voting procedure defined in \cref{def:deepest_voting} with a depth function defined in \cref{def:depth_R} satisfies the properties of Neutrality, Universality and Unanimity. They also study Monotonicity for given depths functions, and they provide a condition on the depth function to satisfy Independence to Irrelevant alternative.

Next, \cite{aubin2022deepest} showed that classical evaluation-based voting rules can be expressed as deepest voting procedures with $L^q$ depths, $q\geq 1$. If evaluation grades are in $\Lambda=\{0,1\}$, as seen in \cite{Dort.et.al}, for all $q\geq 1$, $L^q$ depth voting is equivalent to approval voting \citep{Brams}. If evaluation grades are in $\Lambda=[0,1]$, $L^q$ deepest voting with $q=1$ and $q=2$ recovers respectively majority judgment \citep{Balinski2007} and range voting \citep{Smith}. These results are summarized in Table \ref{tab:deepest_continuous}.

\begin{table}[!ht]
\centering
\begin{tabular}{ll}
\toprule
    \textbf{Depth} & \textbf{Voting process on} [0,1] \\ \midrule
    $L^1$ & Vote to the highest median (majority judgment) \\ 
    $L^2$ & Vote to the highest mean (range voting) \\ \hline 
    & \\ 
    \toprule
    \textbf{Depth} & \textbf{Voting process on} \{0,1\} \\ \midrule
    $L^q$, $q\geq 1$ & Approval voting \\ \hline 
        & \\ 
    \toprule 
    \textbf{Depth} & \textbf{Voting process on} $\S_\dim$ \\ \midrule
    $L^1$ & Bucklin's voting \\ 
    $L^2$ & Borda's voting \\ \bottomrule
\end{tabular}
\caption{(Continuous) $L^q$ deepest voting rules.}
\label{tab:deepest_continuous}
\end{table}

Considering ranking-based procedure, we can also write some procedures as continuous deepest voting process. Indeed, let us consider a framework where, for each voter $v$, the opinions on the candidates are a vector $(e_{cv})_{c=1,\dots,\dim}$ which is a permutation of $\{1,\dots,\dim\}$. The opinions belong to $\E=\S_\dim$, the set of permutations of $\{1,\dots,\dim\}$, which is included in $\{1,\dots,\dim\}^\dim$. Consequently, we can see the rankings as evaluations and apply depth functions defined on $\R^\dim$. We can recover existing voting rules.

In particular, let consider Borda's procedure 
 and Bucklin's procedure. 
\cite{McCabe2006} defines the Bucklin's voting rule (also called Majoritarian Compromise) winner as the candidate with the lowest median rank.
We can prove that these social choice functions are related respectively to $L^2$ and $L^1$ deepest voting.

\begin{proposition}\label{prop:borda_bucklin_continuous}
Consider $n$ voters and $\dim$ candidates. Let $e_{cv}$ be the rank in $\{1,\dots,\dim\}$ given by voter $v$ to candidate $c$, and denote $\Phi=(e_{cv})_{v,c}\in\S_\dim$ the obtained ranks. Then,
\begin{enumerate}[label=(\roman*),topsep=0pt]
    \item Borda's procedure is the (continuous) $L^2$-deepest voting applied on $\Phi$;
    \item Bucklin's winner(s) is included in the (continuous) $L^1$-deepest voting set on $\Phi$, $c_D^*$ defined in \cref{def:deepest_voting}. When the winner is unique Bucklin's winner coincides with the winner of $L^1$ deepest voting.
\end{enumerate}
\end{proposition}
The proof is straightforward and thus omitted. Noted that with $L^1$ procedure, we distinguish whether the set $c_D^*$ is a singleton or not. Indeed, consider the case of 3 candidates and 4 voters with $\Phi$ equal to 
\begin{center}
\begin{tabular}{lcccc}
\toprule
& $v_1$ &$v_2$ &$v_3$ &$v_4$\\\midrule
$c_1$&  2 & 2 & 1 & 1\\
$c_2$&  1 & 1 & 3 & 3\\
$c_3$ &  3 & 3 & 2 & 2\\ 
\bottomrule
  \end{tabular}
\end{center}
In this configuration, Bucklin's winner is candidate $c_1$. But all vector $(a_1,a_2,a_3)^\top$ with $a_1\in[1,2]$, $a_2\in[1,3]$, $a_3\in[2,3]$ is in the deepest set. Hence $c_1$, $c_2$ and $c_3$ belong to $c_D^*$. 

\section{Deepest voting on permutations}
\label{sec:ranks}

Let us now deal with ranking-based voting. The voters give opinions on candidates on the form of rankings, that is, for each voter, the vector of opinions belongs to $\E=\S_\dim$, the set of permutations of $\{1,\dots,\dim\}$. With $n$ voters and $\dim$ candidates, the opinion matrix, denoted $\Phi$, belongs to $\E^n$. In that case, we consider depth functions defined on $\S_\dim$, with properties given in \cref{def:Goibert_depths}. For a given depth function $D()$, the associated voting rule is obtained maximizing $D(.,\Phi)$, as described in \cref{def:deepest_voting}.

In particular, based on \cref{def:depth_Goibert}, we propose to define the $p$-Fréchet mean voting rule on rankings. 

\begin{definition}\label{def:votingrules}
Consider a voting situation $\Phi\in\S_\dim^n$, and a distance $d()$ on $\S_\dim$. For all $p \geq 1$, denote  \[ \S^*(d,\Phi,p):=  
        \argmin_{\sigma \in \S_\dim} U_{d,\Phi,p}(\sigma), \mbox{  with  } U_{d,\Phi,p}(\sigma) = \frac{1}{n}\sum_{v=1}^n d^p(\sigma,\Phi(.,v))
         \]
the set of $p$-Fréchet means of $(\S_\dim,\Phi)$ with respect to $d$. Let
$$c_{p,d}^*:=\{c \in \{1,\ldots, \dim\},\; \sigma^*(c)=1,\; \sigma^* \in \S^*(d,\Phi,p) \}.$$
If $c_{p,d}^*$ is unique, then the winner of the election associated to $(d,p)$ is the candidate $c_{p,d}^*$. If $c_{p,d}^*$ is not unique, there is no unique winner of the election associated to $(d,p)$.
\end{definition}

As for continuous deepest voting, we will show that we recover usual voting rules.

\subsection{Classical rules as Fréchet voting}

Combining distances between permutations, depth functions on permutations lead to a huge number of voting rules, as defined in definition \ref{def:deepest_voting}. Some of these rules correspond to well-known voting rules. 

\subsubsection{Unweighted distances}
\label{sec:unweighted}

Let us first consider distances without weights (or with weights all equal to $1/\dim$).
We can establish links between Kemeny's procedure 
 and Borda's procedure 
 and some deepest voting procedure defined in \cref{def:votingrules}. Recall that Kendall's distance was defined in \eqref{eq:kendall} and that Spearman $\rho$~distance is $q$-Minkowski-Hölder distance (\eqref{eq:minkowski}) with $q=2$.

\begin{proposition}\label{prop:kemenyyoung}
  The deepest voting rule associated to consensus ranking and Kendall distance on $\S_\dim$ elects the Kemeny's winner of the election.
 \end{proposition}

\begin{proposition}\label{prop:borda}
   The deepest voting rule associated to barycenter ranking and Spearman $\rho$ distance is equivalent to Borda's voting method.
\end{proposition}

Proofs of these propositions can be found in \cref{sec:proofs_unweighted}. Note that \cref{prop:borda_bucklin_continuous} shows that Borda's procedure can also be seen as a continuous $L^2$-deepest voting.

Borda's voting method is known to be the ranking-mean voting rule. It is not surprising that it corresponds to the Spearman $\rho$ distance deepest voting. Similarly, Bucklin voting method is known to be the ranking-median voting rule, and Spearman footrule ($q$-Minkowski-Hölder distance \eqref{eq:minkowski} with $q=1$) seems a natural distance to associate with. However, it is not the case, as stated in the following proposition.

\begin{proposition}\label{prop:bucklin}
   Bucklin's voting rule does not correspond to the deepest voting rule associated to consensus ranking and Spearman footrule distance, nor to the deepest voting rules associated to barycenter ranking and Spearman footrule distance as permutation distance.
\end{proposition}

A proof is provided in \cref{sec:proofs_unweighted}. Remark that \cref{prop:borda_bucklin_continuous} expresses Bucklin's procedure as a result of a continuous deepest voting, where ranks are seen as evaluations.

\subsubsection{Weighted distances}
\label{sec:weighted}
Let us suppose now that the distances are weighted, as proposed in \cref{sec:weights}, with non uniform weights. Special set of weights can be used to represent plurality and anti-plurality rules with weighted Hamming distance \eqref{eq:hamming_weighted} or weighted Spearman footrule distance (weighted $q$-Minkowski-Hölder distance \eqref{eq:minkowski_weighted} with $q=1$).

In the following we denote $DV(p, d, W)$ the  voting rule on permutations associated to the $p$-Fréchet mean as depth function, the permutation distance $d()$ and the weights $W$. Recall that, as discussed previously, the associated weighted \emph{distance} functions are not true distances on the set of permutations. For the sake of simplicity, we denote the optimum permutation given by the voting rule as the {\it deepest permutation}.

Denote $W({1,(1)})$ the weights such that $w_{r,r'}=1$ if $r=1$ or $r'=1$ and $w_{r,r'}=0$ elsewhere.
Denote also $W({-1,(m)})$ the weights such that $w_{r,r'}=-1$ if $r=m$ or $r'=m$ and $w_{r,r'}=0$ elsewhere.
These weights allow us to link some well known voting methods to the depth framework on permutations.

\begin{proposition}\label{prop:plurality}
    $DV(1, \mbox{Hamming}, W({1,(1)}))$ are equivalent to  Plurality voting method.
   \end{proposition}

\begin{proposition}\label{prop:antiplurality}
    $DV(1, \mbox{Hamming}, W({-1,(m)}))$ are equivalent to Antiplurality voting method.
\end{proposition}

Proofs are given in \cref{sec:proofs_weighted}.

\begin{table}[!ht]
\centering
\begin{tabular}{llll}
\toprule
    Distance & Weights  & $p$ in $p$-Fréchet mean & Voting rule \\ \midrule
    Kendall & none  & 1 & Kemeny \\ 
    Spearman $\rho$ & none & 2 & Borda \\ 
    Hamming & $W(1,(1))$ & 1 & Plurality \\
    Hamming & $W(-1,(\dim))$ & 1 & Antiplurality \\
\bottomrule
\end{tabular}
\caption{Connections between deepest voting rules on permutations and classical voting rules.}
\label{tab:deepest_ranks}
\end{table}

Obtained links between classical voting rules and deepest voting rules are summarized in \cref{tab:deepest_ranks}.
The objective of this paper is not to provide an exhaustive list of links between usual social choice functions and deepest voting procedures, but to highlight the connections. This works shows that several social choice functions can be associated with a depth function, and hence write as an optimization problem.

\subsection{Voting rules properties}
\label{sec:properties}

Studying formal properties of voting rules is a matter of interest for an axiomatic approach of elections. \citep[Chapter 2]{Felsenthal2018} contains a good review of classical properties of a voting rule. We propose hereafter to discuss some of them. \cite{aubin2022deepest} has related some properties of deepest voting rules with the mathematical properties of the continuous depth functions. We extend their result to permutation-based deepest voting.

\begin{definition}[Voting rules properties]~
\label{def:properties}
\begin{itemize}[topsep=0pt]
\item \emph{Neutrality}. 
The social choice function gives the same result by permuting the rows of $\Phi$ (\emph{i.e.} by permuting the candidates).
\item {\emph{Anonymity}.}  
The voting procedure gives the same result by permuting the columns of $\Phi$ (\emph{i.e.} by permuting the voters).
\item  {\emph{Universality}.}  For all $\Phi \in \S_m^n $, the voting procedure provides a subset of $\C$.
\item {\emph{Unanimity}.}
If a candidate is in first position for all voters, it should be in first position in the voting procedure's results. 
\item {\emph{Monotonicity}.}  
Suppose that a voting situation $\Phi$ gives a winner $c^*$, and that there exists a voter $v_0\in\{1,\dots,n\}$ such that $\Phi(c^*,v_0)= \alpha \neq 1$. Let $c_0$ be the candidate such that $\Phi(c_0,v_0)=\alpha-1$. 
Consider another voting situation $\widetilde{\Phi}$ such that $\widetilde{\Phi}=\Phi$ except that $\widetilde{\Phi}(c^*,v_0)=\alpha-1$ and $\widetilde{\Phi}(c_0,v_0)=\alpha$. Then the winner of the voting procedure is still $c^*$.
\item {\emph{Independence to Losers}.} 
  Consider a voting situation $\Phi$ with $\dim$ candidates $\C=\{c_1,\dots,c_\dim\}$ and $n$ voters, with a winner $c^*$. Consider another voting situation $\widetilde{\Phi}$ with $\dim-1$ candidates and $n$ voters, where a candidate $c_0\neq c^*$ as been removed. Suppose that $\widetilde{\Phi}$ is identical on the reduced set of candidates. That is, for any voter $v \in \{1, \ldots, n\}$, the order of $\{\Phi(c,v),\ c\in\C\setminus\{c_0\}\}$ is the same as the order of $\{\widetilde{\Phi}(c,v),\ c\in\C\setminus\{c_0\}\}$. Then the winner for $\widetilde{\Phi}$ is still $c^*$.
\item {\emph{Condorcet Winner}}. A candidate $c_0\in \C$ is a Condorcet winner if for all candidates $c \in \C\setminus\{c_0\}$, the number of voters $v$ such that $\Phi(c_0,v)<\Phi(c,v)$ is strictly higher than the number of voters $v$ such that $\Phi(c_0,v)>\Phi(c,v)$.
A voting rule satisfies the Condorcet Winner property if it elects the Condorcet winner when it exists.  
\item {\emph{Condorcet Looser}}.  
 A candidate $c_0\in \C$ is a Condorcet looser if for all candidates $c \in \C\setminus\{c_0\}$, the number of voters $v$ such that $\Phi(c_0,v)<\Phi(c,v)$ is strictly lower than the number of voters $v$ such that $\Phi(c_0,v)>\Phi(c,v)$. 
A voting rule satisfies the Condorcet looser property if it never elects the Condorcet looser when it exists.  

\end{itemize}
\end{definition}

Our objective in this section is to study if the properties stated above are satisfied with deepest voting rules. As said previously, some results were obtained in \cite{aubin2022deepest} and we concentrate here on permutation-based deepest voting.  

Proofs of this section can be found in \cref{sec:proofs_properties}.

First, Neutrality, Anonymity and Universality are satisfied whatever the depth function used. As for Unanimity, we can prove it is satisfied by deepest voting rules for a large category of permutation-based depth functions.

\begin{proposition}\label{prop:NAU}
   Any deepest voting process satisfies Neutrality, Anonymity and Universality.
\end{proposition}

 \begin{proposition}\label{prop:una}

 Consider a distance on permutations which is either Hamming, Kendall or a $q$-Minkowski-Hölder distance on permutations, with $q \in  [1, +\infty[ $. Then, for all $p\geq 1$, the associated $p$-Fréchet mean voting satisfies Unanimity.
 
 When $q=\infty$, $p$-Fréchet mean voting satisfies Unanimity when the winner is unique, but it may not satisfy Unanimity otherwise  (\emph{i.e.} there are at least two permutations with a different top-ranked candidate in the deepest set).

\end{proposition}

The Condorcet-winner property and the Monotonicity are satisfied by some deepest voting procedures, as stated in the two following propositions.

\begin{proposition} \label{prop:condowinner}
~
\begin{enumerate}[topsep=-\baselineskip]
\item For $p=1$, Kendall-based consensus ranking satisfies the Condorcet-winner property.
\item For $p>1$, Kendall-based $p$-Fréchet mean does not satisfy the Condorcet-winner property.
\item For all $p \geq 1$, Cayley and Hamming based voting rules do not satisfy the Condorcet-winner property.
\item For $p = 1$ and all $q \geq 1$, $q$-Minkowski-Hölder based voting rules do not satisfy the Condorcet-winner property.
\end{enumerate}
\end{proposition}

\begin{proposition} \label{prop:monotonie}
~
\begin{enumerate}[topsep=-\baselineskip]
    \item  The consensus ranking ($p=1$) associated to the Spearman footrule ($1$-Minkowski-Hölder) satisfies the Monotonicity property.
    \item  The consensus ranking ($p=1$) associated to the Hamming distance does not satisfy the Monotonicity property.
\end{enumerate}

\end{proposition}

Finally, Independence to Losers property is not satisfied by a large category of permutation-based voting rules.

\begin{proposition} \label{prop:IIA}
Hamming-based, Cayley-based, Kendall-based and Minkowski-Hölder based consensus ($p=1$) deepest voting rules do not satisfy Independence to Losers property.
\end{proposition}

\cref{tab:recap_properties} below summarizes the results of this section.

\begin{table}[!ht]
{\small \hspace{-0.75cm}
\begin{tabular}{lcccccc}\toprule
Distance &  Neutrality & Anonymity & Unanimity & Monotonicity & Ind.Losers        & Condorcet winner \\ 
\midrule
Hamming &   $\checked$  & $\checked$ & $\checked$ & N for $p=1$ & N for $p=1$ & N  \\
Kendall &   $\checked$  & $\checked$ & $\checked$ & .           & N           & $\checked$ only for $p=1$  \\
Cayley &  $\checked$  & $\checked$ & .            & .           & N           & N\\ \hline
Minkowski-Hölder &&&&&& \\
  \hspace{0.5cm } q=1   &   $\checked$  & $\checked$ & $\checked$ &$\checked$  for $p=1$ & N & N for $p=1$ \\
\hspace{0.5cm } $q>1$    &   $\checked$ & $\checked$ & $\checked$ & . & N & N for $p=1$ \\ \bottomrule
\end{tabular}
}
\caption{The table synthesizes if the deepest voting procedure associated to several depth functions the properties given in \cref{def:properties}. Deepest voting rules are $p$-Fréchet means, with different distances on permutations. $\checked$ corresponds to verified properties, N to non verified ones, . to cases where no result was obtained. When results are only obtained for a given value of $p$, it is precised in the corresponding entry.}
\label{tab:recap_properties}
\end{table}

\section{Conclusion}

In this paper, we extended the deepest voting framework to ranking-based voting rules by introducing depth functions defined on the space of permutations. This reinforces links between social choice theory and statistical depth. This approach provides a unified statistical interpretation of several classical social choice procedures as solutions to optimization problems based on distances between rankings.

Note that we focused in this work on distance-based depth functions. Other depth functions can be used, such as half-space or projection depths, that rely more on the geometrical properties of the data. It may be considered in future work.

We showed that well-known rules such as Kemeny, Borda, Plurality and Antiplurality can be expressed as deepest voting procedures associated with appropriate (weighted) distances and Fréchet means. Moreover, we analyzed key axiomatic properties (universality, monotonicity, etc.) of deepest voting procedures.

Beyond these connections, the proposed framework opens the way to the design of new voting rules based on weighted or problem-specific distances between rankings. Future work will focus on a deeper axiomatic analysis or computational aspects.

\section{Acknowledgement}

This work has been supported by ANR France 2030 agency through the PEPR Maths-Vives Condorcet project ANR-24-EXMA-0001.

\bibliography{biblio}

\begin{thebibliography}{25}
\providecommand{\natexlab}[1]{#1}
\providecommand{\url}[1]{\texttt{#1}}
\expandafter\ifx\csname urlstyle\endcsname\relax
  \providecommand{\doi}[1]{doi: #1}\else
  \providecommand{\doi}{doi: \begingroup \urlstyle{rm}\Url}\fi

\bibitem[Aubin et~al.(2022)Aubin, Gannaz, Leoni, and Rolland]{aubin2022deepest}
J.-B. Aubin, I.~Gannaz, S.~Leoni, and A.~Rolland.
\newblock Deepest voting: a new way of electing.
\newblock \emph{Mathematical Social Sciences}, 116:\penalty0 1--16, 2022.

\bibitem[Balinski and Laraki(2007)]{Balinski2007}
M.~Balinski and R.~Laraki.
\newblock A theory of measuring, electing, and ranking.
\newblock \emph{Proceedings of the National Academy of Sciences}, 104\penalty0
  (21):\penalty0 8720--8725, 2007.

\bibitem[Brams and Fishburn(2007)]{Brams}
S.~Brams and P.~C. Fishburn.
\newblock \emph{Approval voting}.
\newblock Springer, 2007.

\bibitem[Deza and Deza(2016)]{deza2016encyclopedia}
M.~Deza and E.~Deza.
\newblock \emph{Encyclopedia of Distances}.
\newblock Springer Berlin Heidelberg, 2016.
\newblock ISBN 9783662528440.
\newblock URL \url{https://books.google.fr/books?id=KQHdDAAAQBAJ}.

\bibitem[Deza and Huang(1997)]{dist.permut}
M.~Deza and T.~Huang.
\newblock Metrics on permutations, a survey.
\newblock \emph{J. Comb. Inf. Sys. Sci.}, 23, 02 1997.

\bibitem[Diaconis(1988)]{Diaconis88}
P.~Diaconis.
\newblock Group representations in probability and statistics.
\newblock \emph{Lecture Notes-Monograph Series}, 11:\penalty0 i--192, 1988.
\newblock ISSN 07492170.
\newblock URL \url{http://www.jstor.org/stable/4355560}.

\bibitem[Dort and friends(2025)]{Dort.et.al}
L.~Dort and friends.
\newblock Approval voting.
\newblock \emph{to be written}, 2025.

\bibitem[Elkind et~al.(2009)Elkind, Faliszewski, and Slinko]{ElkindTARK}
E.~Elkind, P.~Faliszewski, and A.~Slinko.
\newblock On distance rationalizability of some voting rules.
\newblock In \emph{Proceedings of the 12th Conference on Theoretical Aspects of
  Rationality and Knowledge}, TARK '09, page 108–117, New York, NY, USA,
  2009. Association for Computing Machinery.

\bibitem[Elkind et~al.(2010)Elkind, Faliszewski, and Slinko]{Elkind2010OnTR}
E.~Elkind, P.~Faliszewski, and A.~M. Slinko.
\newblock On the role of distances in defining voting rules.
\newblock In \emph{Adaptive Agents and Multi-Agent Systems}, 2010.
\newblock URL \url{https://api.semanticscholar.org/CorpusID:11697212}.

\bibitem[Felsenthal and Nurmi(2018)]{Felsenthal2018}
D.~S. Felsenthal and H.~Nurmi.
\newblock \emph{Voting Procedures for Electing a Single Candidate}.
\newblock SpringerBriefs in Economics. Springer, 2018.
\newblock \doi{10.1007/978-3-319-74033-1_3}.

\bibitem[Felsenthal and Nurmi(2019)]{Felsenthal}
D.~S. Felsenthal and H.~Nurmi.
\newblock \emph{20 Voting Procedures Designed to Elect a Single Candidate},
  pages 5--16.
\newblock Springer International Publishing, Cham, 2019.
\newblock ISBN 978-3-030-12627-8.
\newblock \doi{10.1007/978-3-030-12627-8_2}.

\bibitem[Fr\'echet(1948)]{Frechet}
M.~Fr\'echet.
\newblock Les \'el\'ements al\'eatoires de nature quelconque dans un espace
  distanci\'e.
\newblock \emph{Annales de l'institut Henri Poincar\'e}, 10\penalty0
  (4):\penalty0 215--310, 1948.
\newblock URL \url{http://www.numdam.org/item/AIHP_1948__10_4_215_0/}.

\bibitem[Goibert et~al.(2022)Goibert, Clemencon, Irurozki, and
  Mozharovskyi]{Goibert}
M.~Goibert, S.~Clemencon, E.~Irurozki, and P.~Mozharovskyi.
\newblock Statistical depth functions for ranking distributions: Definitions,
  statistical learning and applications.
\newblock In \emph{International Conference on Artificial Intelligence and
  Statistics}, pages 10376--10406. PMLR, 2022.

\bibitem[Kumar and Vassilvitskii(2010)]{weighted_distances}
R.~Kumar and S.~Vassilvitskii.
\newblock Generalized distances between rankings.
\newblock In \emph{Proceedings of the 19th International Conference on World
  Wide Web}, WWW '10, page 571–580, New York, NY, USA, 2010. Association for
  Computing Machinery.
\newblock ISBN 9781605587998.
\newblock \doi{10.1145/1772690.1772749}.
\newblock URL \url{https://doi.org/10.1145/1772690.1772749}.

\bibitem[Mahajne et~al.(2015)Mahajne, Nitzan, and Volij]{MahajneNV15}
M.~Mahajne, S.~Nitzan, and O.~Volij.
\newblock Level {\textbackslash}(r{\textbackslash}) consensus and stable social
  choice.
\newblock \emph{Soc. Choice Welf.}, 45\penalty0 (4):\penalty0 805--817, 2015.

\bibitem[McCabe-Dansted and Slinko(2006)]{McCabe2006}
J.~McCabe-Dansted and A.~Slinko.
\newblock Exploratory analysis of similarities between social choice rules.
\newblock \emph{Group Decision and Negotiation}, 15:\penalty0 77--107, 12 2006.

\bibitem[Meskanen and Nurmi(2008)]{Meskanen2008}
T.~Meskanen and H.~Nurmi.
\newblock \emph{Closeness Counts in Social Choice}, pages 289--306.
\newblock Springer Berlin Heidelberg, Berlin, Heidelberg, 2008.

\bibitem[Mosler(2013)]{mosler2013depth}
K.~Mosler.
\newblock Depth statistics.
\newblock \emph{Robustness and complex data structures: Festschrift in Honour
  of Ursula Gather}, pages 17--34, 2013.

\bibitem[Nitzan(1981)]{Nitzan81}
S.~Nitzan.
\newblock Some measures of closeness to unanimity and their implications.
\newblock \emph{Theory and Decision}, 13\penalty0 (2):\penalty0 129--138, 1981.

\bibitem[Smith(2000)]{Smith}
W.~D. Smith.
\newblock Range voting.
\newblock \url{http://rangevoting.org/RangeVoting.html}, 2000.
\newblock Accessed: 2014-10-12.

\bibitem[Tukey(1975)]{Tukey}
J.~W. Tukey.
\newblock Mathematics and the picturing of data.
\newblock In \emph{Proceedings of the International Congress of Mathematicians,
  Vancouver, 1975}, volume~2, pages 523--531, 1975.

\bibitem[Young and Levenglick(1978)]{young1978consistent}
H.~P. Young and A.~Levenglick.
\newblock A consistent extension of condorcet’s election principle.
\newblock \emph{SIAM Journal on applied Mathematics}, 35\penalty0 (2):\penalty0
  285--300, 1978.

\bibitem[Zuo(2003)]{Zuo2003}
Y.~Zuo.
\newblock Projection-based depth functions and associated medians.
\newblock \emph{Annals of Statistics}, 31\penalty0 (5):\penalty0 1460--1490,
  2003.

\bibitem[Zuo(2004)]{Zuo2004}
Y.~Zuo.
\newblock Robustness of weighted {$L^p$}--depth and {$L^p$}--median.
\newblock \emph{Allgemeines Statistisches Archiv}, 88\penalty0 (2):\penalty0
  215--234, 2004.

\bibitem[Zuo and Serfling(2000)]{ZuoSerfling}
Y.~Zuo and R.~Serfling.
\newblock General notions of statistical depth function.
\newblock \emph{Annals of statistics}, pages 461--482, 2000.

\end{thebibliography}
\bibliographystyle{abbrvnat}

\appendix

\section{Proofs for unweighted distances}
\label{sec:proofs_unweighted}

This section provides the proofs of \cref{sec:unweighted}, dealing with unweighted deepest voting on permutations.

\subsection*{Proof of \cref{prop:kemenyyoung}}

Quoting  \cite{Felsenthal}, p. 26, \begin{quote}Kemeny‘s
social choice procedure can also be viewed as finding the most likely (or the best predictor, or the best compromise) true social
preference ordering, called the median preference ordering, i.e., that social preference ordering $S$ that minimizes the
sum, over all voters $i$, of the number of pairs of candidates that are ordered oppositely by $S$ and by the $i^{th}$ voter.\end{quote} It is exactly the definition of the consensus ranking deepest voting rule with Kendall distance.

\subsection*{Proof of \cref{prop:borda}.}

We consider here that the Borda voting rule elects a unique candidate. 

Suppose that $c^*$ is chosen with the Borda voting rule. It means that 
\begin{equation}
    \label{eq:borda}
\forall c\in\mathcal{C},\ c\neq c^*, \ \sum_{v=1}^n\sigma_{v}(c^*)<\sum_{v=1}^n\sigma_{v}(c),
\end{equation}
 where $\sigma_v(c)=e_{cv}$ denotes the rank given by voter $v$ to candidate $c$. Then $\{\sigma_{v})_{c}\}$ is a finite family of cardinal $n\ge2$. Denote $\Pi_n$ its empirical distribution.

Let $\sigma^*\in\S_\dim$ a permutation and $c^* \in \{1, \ldots, m\}$ such that $\sigma^*(c^*)=1$.

Let $c_0\in\mathcal{C}$, $c_0\neq c^*$. Note $\tau=(c^*,c_0)$ the transposition exchanging $c^*$ and $c_0$. Let us show that $\sigma^*$ has a lower 2-Frechet mean than $\sigma^* \circ \tau$. That is, we want to prove that 
\[\mathbb{E}[d^2(\sigma^*,\Pi_n)]<\mathbb{E}[d^2(\sigma^* \circ \tau,\Pi_n)],\]
with $d()$ the $q$-Minkowski-Holder distance for $q=2$. 

Using successively the definition of $\Pi_n$ and the definition of the distance $d()$, we have \begin{align*}
   {\mathbb{E}[d^2(\sigma^*,\Pi_n)]-\mathbb{E}[d^2(\sigma^* \circ \tau,\Pi_n)]} 
     & =  \frac{1}{n}\sum_{v=1}^nd^2(\sigma^*,\sigma_{v})-\frac{1}{n}\sum_{v=1}^nd^2(\sigma^* \circ  \tau,\sigma_{v}) \\ 
    &=  \frac{1}{n}  \sum_{v=1}^n \sum_{c=1}^\dim  (\sigma^*(c)-\sigma_{v}(c))^{2} -\frac{1}{n}  \sum_{v=1}^n  \sum_{c=1}^\dim (\sigma^*\circ \tau (c)-\sigma_{v}(c))^{2}.
\end{align*}
Definition of $\tau$ yields  
\begin{align*}
   \lefteqn{\mathbb{E}[d^2(\sigma^*,\Pi_n)]-\mathbb{E}[d^2(\sigma^* \circ \tau,\Pi_n)]} \\
 &=   \frac{1}{n}\sum_{v=1}^n \left[  (\sigma^*(c^*)-\sigma_{v}(c^*))^{2}+(\sigma^*(c_0)-\sigma_{v}(c_0))^{2} 
 -(\sigma^*(c_0)-\sigma_{v}(c^*))^{2}-(\sigma^*(c^*)-\sigma_{v}(c_0))^{2}  \right] \\
    & = \frac{2}{n}\sum_{v=1}^n (\sigma^*(c^*)-\sigma^*(c_0))( \sigma_{v}(c_0)-\sigma_{v}(c^*)) \\
    & = \frac{2}{n}(\sigma^*(c^*)-\sigma^*(c_0))\Bigl(\sum_{v=1}^n \sigma_{v}(c_0)-\sum_{v=1}^n\sigma_{v}(c^*)\Bigr).
\end{align*}
Using Equation \eqref{eq:borda} and the fact that $\sigma^*$ was defined such that $\sigma^*(c^*)-\sigma^*(c_0)<0$, we deduce that the right hand side is strictly negative.
Hence the deepest permutations are such that $\sigma^*(c^*)=1$, and therefore the deepest voting rule elects the winner of Borda rule.

\subsection*{Proof of \cref{prop:bucklin}}

Consider a case with 5 voters and 4 candidates, with the voting table is proposed in \cref{tab:propbucklin} below.
\begin{table}[!ht]
    \centering
    \begin{tabular}{cccccc} \toprule
            & $v_1$ & $v_2$ & $v_3$ & $v_4$ & $v_5$ \\ \midrule
        $c_1$ &   1 & 1 & 4 & 4 & 3\\
        $c_2$ &  2 & 2 & 2 & 2 & 2\\
        $c_3$ &  3 & 3 & 3 & 3 & 1\\
        $c_4$ &  4 & 4  & 1 & 1 & 4 \\ \bottomrule
    \end{tabular}
    \caption{Matrix $\Phi$ of the voting situation for the proof of \cref{prop:bucklin}.}
    \label{tab:propbucklin}
\end{table}
In this setting, the median rankings are (3,2,3,4) and therefore $c_2$ is the winner of Bucklin's vote as it has the smallest median ranking.
But the optimum of the $p$-Fréchet mean with the Spearman footrule distance is $(1,2,3,4)^\top$, for both cases $p=1$ and $p=2$. Therefore $c_1$ is the winner of $DV(1, \mbox{Spearman footrule})$, and $DV(2, \mbox{Spearman footrule})$. This concludes the proof.

\section{Proofs for weighted distances}
\label{sec:proofs_weighted}

The proofs dealing with results on depths based on weighted permutation distance, given in \cref{sec:weighted}, are displayed in this section.

\subsection*{Proof of \cref{prop:plurality}}

    The weighted Hamming distance with the weights $W(1,(1))$ writes as
    $$ \text{for all }\sigma,\tau\in\S_\dim,\; d(\sigma, \tau)=\begin{cases}
    2 & \text{if } \sigma^{-1}(1) \neq \tau^{-1}(1),\\
    0 & \text{if } \sigma^{-1}(1) = \tau^{-1}(1).\end{cases}$$ Therefore, $\sigma^* \in \S_\dim$ minimizes $\sum_{v=1}^n d(\sigma, \sigma_v)$ if  $$\sigma^{*-1}(1)= \argmax_{c = 1, \ldots, \dim} \sum_{v=1}^n \1(\sigma^{-1}_v(1)=c).$$ It concludes the proof since it is exactly the definition of Plurality voting rule.

\subsection*{Proof of \cref{prop:antiplurality}}

    The weighted Hamming distance with the weights $W({-1,(\dim)})$ writes as
    $$ \text{for all }\sigma,\tau\in\S_\dim,\; d(\sigma, \tau)=\begin{cases}
    -1 & \text{if } \sigma^{-1}(\dim) \neq \tau^{-1}(\dim),\\
    0 & \text{if } \sigma^{-1}(\dim) = \tau^{-1}(\dim).\end{cases}$$
    Therefore, $\sigma^* \in \S_\dim$ minimizes $\sum_{v=1}^n d(\sigma^*, \sigma_v)$, if  $$\sigma^{*-1}(\dim)= \argmin_{c = 1, \ldots, \dim} \sum_{v=1}^n \1((\sigma_v^{-1}(c)=\dim)).$$  It concludes the proof since it is exactly the definition of Antiplurality voting rule.

\section[Proofs of voting properties]{Proofs of \cref{sec:properties}}
\label{sec:proofs_properties}
 
 This section provides the proofs dealing with the properties of the deepest voting procedures, stated in \cref{sec:properties}. The properties are defined in \cref{def:properties}.
   
\subsection*{Proof of \cref{prop:NAU}}

\begin{description}[font=\bf]
\item[\emph{Anonymity}.]  For continuous depth functions, the property has been proven in \cite[Proposition 1]{aubin2022deepest}. Consider a depth on permutations, defined in \cref{def:depth_Goibert}. For any voting situation $\Phi$, with empirical distribution $\Pi_n$, the deepest voting relates only on the function which associate a permutation $\sigma\in\S_m$ to $U_{d,\Pi_n,p}(\sigma)=\frac{1}{n}\sum_{v=1}^n d^p(\sigma,\Phi(.,v))$.
It is straightforward that, $U_{d,\Pi_n,p}()$ is not modified by permuting the columns of $\Phi$ (\emph{i.e.} permuting the voters), and hence the deepest voting procedure remains identical.  
\item[\emph{Neutrality}.] 
For continuous depth functions, the property has been proven in \cite[Proposition 1]{aubin2022deepest}. Consider a depth on permutations, defined in \cref{def:depth_Goibert}. Similarly, for any voting situation $\Phi$, with empirical distribution $\Pi_n$, the deepest voting relates only on the function which associate a permutation $\sigma\in\S_m$ to $U_{d,\Pi_n,p}(\sigma)=\frac{1}{n}\sum_{v=1}^n d^p(\sigma,\Phi(.,v))$.
Permuting the rows of $\Phi$ (\emph{i.e.} permuting the candidates) by a permutation $\tau$, $d(\tau\circ\sigma,\Phi(\tau(.),v))=d(\sigma,\Phi(.,v))$. Hence the deepest voting procedure provides the same argmax set, up to the permutation $\tau$, and consequently the same candidate will be winning.
\item[\emph{Universality}] Universality is obvious as the definition domain $\E_D$ of the depth functions is $\dim$-dimensional.
\end{description}

\subsection*{Proof of \cref{prop:una}}

Recall that Unanimity means that if there exists $c^* \in \{1, \ldots, \dim\}$ such that $\forall v \in \{1,\ldots n\}$, $\Phi(c^*,v)=1$, then $\sigma^*(c^*)=1$, with $\sigma^*$ the deepest permutation. 

Suppose that all voters agree on the same best candidate, \emph{i.e.} suppose that   $\exists c^* \in \{1, \ldots, \dim\}$ such that $\forall v \in \{1,\ldots n\}$, $\Phi(c^*,v)=1$. 

Let $\sigma\in\S_\dim$ such that $\sigma(c^*)\neq 1$. Without loss of generality, consider $c^\ast=1$ and $\sigma(2)=1$. Consider $\tau$ the permutation such that $\tau(1)=1$, $\tau(2)=\sigma(1)$, and $\tau(c)=\sigma(c)$ if $c=3,\dots,\dim$. Let us prove that then $\tau$ as a lower $p$-Fréchet mean than $\sigma$.

Denote
$$\Delta:=\frac{1}{n}\sum_{v=1}^n \Bigl(d^p(\sigma,\sigma_v)-d^p(\tau,\sigma_v)\Bigr)$$
the difference of the $p$-Fréchet means of $\sigma$ and $\tau$. Let us prove that $\Delta\geq 0$. More precisely, we are going to prove that for all $v\in\{1,\dots,n\}$ we have $d(\sigma,\sigma_v)>d(\tau,\sigma_v)$.

Let us distinguish with respect to the distance.

\begin{description}

    \item[For the Hamming distance.] Let $v\in\{1,\dots,n\}$. Since $\sigma(1)\neq\sigma_v(1)$, $\sigma_2\neq\sigma_v(2)$, and $\tau(1)\neq\sigma_v(1)$, we have \begin{align*}
        d(\sigma, \sigma_v) &= 2 + \sum_{c=3}^\dim \1\{\sigma(c) \neq \sigma_v(c)\},\\
        d(\tau, \sigma_v) & = \1\{\sigma(1) \neq \sigma_v(2)\} + \sum_{c=3}^\dim \1\{\sigma(c) \neq \sigma_v(c)\}.
    \end{align*} 
    With these equations, it is straightforward that $d(\tau,\sigma_v)<d(\sigma,\sigma_v)$.

    \item[For the Kendall distance.] 
    Let $v\in\{1,\dots,n\}$. By definition of the Kendall distance, 
    $$d(\sigma,\sigma_v)=
    \sum_{c=1}^{\dim-1}\sum_{c'=c+1}^{\dim}\1\{(\sigma(c)-\sigma(c'))\,(\sigma_v(c)-\sigma_v(c'))< 0\}.$$
    Using the properties of $\sigma$ and $\tau$, and decomposing the sum on candidates with respect to candidate $c=1$, candidate $c=2$ and candidates $c\geq 3$,
    \begin{align*}
        d(\sigma,\sigma_v) & =
    \sum_{c'=2}^{\dim}\1\{(\sigma(1)-\sigma(c'))> 0\} +
    \sum_{c'=3}^{\dim}\1\{(\sigma_v(2)-\sigma_v(c'))> 0\} \\
    &\quad +
    \sum_{c=3}^{\dim-1}\sum_{c'=c+1}^{\dim}\1\{(\sigma(c)-\sigma(c'))\,(\sigma_v(c)-\sigma_v(c'))< 0\},\\
    d(\tau,\sigma_v) & = \sum_{c'=3}^{\dim}\1\{(\sigma(1)-\sigma(c'))(\sigma_v(2)-\sigma_v(c'))< 0\} \\
    &\quad+
    \sum_{c=3}^{\dim-1}\sum_{c'=c+1}^{\dim}\1\{(\sigma(c)-\sigma(c'))\,(\sigma_v(c)-\sigma_v(c'))< 0\}.
    \end{align*}
    Hence $d(\sigma,\sigma_v)-d(\tau,\sigma_v)=\sum_{c'=2}^{\dim}\1\{(\sigma(1)-\sigma(c'))> 0\} + \delta_v$ with \begin{align*}
           \delta_v & = \sum_{c'=3}^{\dim}\Bigl(\1\{(\sigma(1)-\sigma_v(c'))> 0\}+\1\{(\sigma_v(2)-\sigma_v(c'))> 0\} \\ 
           &\qquad\qquad - \1\{(\sigma(1)-\sigma(c'))(\sigma_v(2)-\sigma_v(c'))< 0\}\Bigr).
        \end{align*}
    It is straightforward that each term of the sum above cannot be negative (since if $\1\{(\sigma(1)-\sigma_v(c'))> 0\}+\1\{(\sigma_v(2)-\sigma_v(c'))> 0\}=0$, then $\1\{(\sigma(1)-\sigma(c'))(\sigma_v(2)-\sigma_v(c'))< 0\}=0$). It results that $$d(\sigma,\sigma_v)-d(\tau,\sigma_v)\geq \sum_{c'=2}^{\dim}\1\{(\sigma(1)-\sigma(c'))> 0\}.$$   
    Since $\sigma(1)>\sigma(2)=1$, we deduce that  $d(\sigma,\sigma_v)>d(\tau,\sigma_v)$.

\item[For $q$-Minkowski-Hölder distances, $1\leq q<\infty$.]
Let $v\in\{1,\dots,n\}$. 
First observe that:
\begin{align*}
    d_q^q(\sigma,\sigma_v)&=\sum_{c=1}^\dim |\sigma(c)-\sigma_v(c)|^q  \\
    &=|\sigma(1)-1|^q+|1-\sigma_v(2)|^q+\sum_{c=3}^\dim |\sigma(c)-\sigma_v(c)|^q.
\end{align*}
Similarly,
\begin{align*}
    d_q^q(\tau,\sigma_v)&=|\sigma(1)-\sigma_v(2)|^q+\sum_{c=3}^\dim |\sigma(c)-\sigma_v(c)|^q.
\end{align*}
Hence, $$d_q^q(\sigma,\sigma_v)-d_q^q(\tau,\sigma_v)=|\sigma(1)-1|^q+|1-\sigma_v(2)|^q-|\sigma(1)-\sigma_v(2)|^q.$$ Next,
\begin{itemize}
    \item if $1<\sigma(1)\leq \sigma_v(2)$, $d_q^q(\sigma,\sigma_v)-d_q^q(\tau,\sigma_v)\geq (\sigma(1)-1)^q>0$;
    \item if $1<\sigma_v(2)\leq \sigma(1)$, $d_q^q(\sigma,\sigma_v)-d_q^q(\tau,\sigma_v)\geq (\sigma_v(2)-1)^q>0$.
\end{itemize} 
Hence, $d_q(\sigma,\sigma_v)>d_q(\tau,\sigma_v)$.
This concludes the proof.

\item[For $q$-Minkowski-Hölder distances, $q=\infty$.] \hfill \break
Here we show that when there is unique solution $\sigma^*$ then $p$-Fréchet mean of $\infty$-Minkowski-Hölder satisfies Unanimity and otherwise that's not always the case. \newline
Let's recall that :
$$\Delta:=\frac{1}{n}\sum_{v=1}^n \Bigl((\max_{c} |\sigma(c)-\sigma_v(c)|)^p-(\max_{c} |\tau(c)-\sigma_v(c)|)^p\Bigr).$$
Then 
    \begin{align} \Delta= & \frac{1}{n}\sum_{v=1}^n \Bigl((\max(\sigma(1)-\sigma_v(1),\sigma(2)-\sigma_v(1), \max_{c>2} |\sigma(c)-\sigma_v(c)|))^p \\ & -(\max(\tau(1)- \sigma_v(1),\tau(2)-\sigma_v(2), \max_{c>2} |\tau(c)-\sigma_v(c)|))^p\Bigr).\\
\end{align}
Since $\sigma(2)=1$, $\forall v, \sigma_v(1)=1$, $\tau(1)=\sigma(2)=1$ and $\tau(2)=\sigma(1)>1$, 
$$\Delta=\frac{1}{n}\sum_{v=1}^n \Bigl((\max(\sigma(1)-1,\sigma_v(2)-1, \max_{c>2} |\sigma(c)-\sigma_v(c)|))^p-(\max(\tau(2)-\sigma_v(2), \max_{c>2} |\tau(c)-\sigma_v(c)|))^p\Bigr).$$
Now, considering that $\max_{c>2} |\sigma(c)-\sigma_v(c)| = \max_{c>2} |\tau(c)-\sigma_v(c)|$ and $\sigma(1)-1 > \tau(2)-\sigma_v(2)$, it leads that for all $p \geq 1$
$$ (\max(\sigma(1)-1,\sigma_v(2)-1, \max_{c>2} |\sigma(c)-\sigma_v(c)|))^p \geq (\max(\tau(2)-\sigma_v(2), \max_{c>2} |\tau(c)-\sigma_v(c)|)),$$
and $\Delta \geq 0$. Note that it doesn't mean that $\Delta \neq 0$ but, necessarily a ranking with $c^*$ in first place minimises the $p$-Fréchet mean for the $L^\infty$ distance. Thus, if the solution is unique then that is this ranking. Otherwise, counter-example below shows that there is no winner therefore unanimity property is not satisfied.

Consider a setting with 5 voters and 6 candidates, with the voting preferences given are according to \cref{tab:condorcetLinfini}. 

\begin{minipage}{\linewidth}
    \centering
    \begin{tabular}{cccccc} \toprule
            & $v_1$ & $v_2$ & $v_3$ & $v_4$ & $v_5$   \\ \midrule
        $c_1$ &   1 & 1 & 1 & 1 & 1\\
        $c_2$ &  4 & 6 & 4 & 2 & 6\\
        $c_3$ &  6 & 4 & 6 & 4 & 5\\
        $c_4$ &   3 & 5 & 3 & 6 & 4\\
        $c_5$ &  5 & 2 & 5 & 5 & 3\\
        $c_6$ &  2 & 3 & 2 & 3 & 2\\ \bottomrule
    \end{tabular}
    \captionof{table}{Counter-example for the proof that $p$-Fréchet mean of $\infty$-Minkowski-Hölder distance may not satisfy unanimity if there are many solutions.}
    \label{tab:condorcetLinfini}
\end{minipage}

The 1-Fréchet means are (1,4,6,3,5,2), (1,5,6,3,4,2) and (2,5,6,3,4,1). 
Both the 2-Fréchet means and the 3-Fréchet means are (1,5,6,3,4,2) and (2,5,6,3,4,1). Hence, the winner is not unique, and candidate $c_6$ is also a preferred candidate in the deepest set.

\end{description}

\subsection*{\bf Proof of \cref{prop:condowinner}}

\begin{enumerate}
\item \textbf{Kendall -based $p$-Fréchet  for $p=1$} \newline
Let $ \Phi=(\sigma_v)_{v=1,\dots,n} \in \S_m^n$ be the set of rankings given by $n$ voters on $\dim$ candidates.
Let $\sigma^*$ be a median permutation of $\Phi$ based on the Kendall distance $d_K$. That is, we consider consensus ranking, with $\sigma^*\in \argmin_{\sigma  \in \S_\dim} \ \sum_{v=1}^n d_K(\sigma,\sigma_v)$.

Suppose that $c_0$ is a Condorcet winner of the election and that $c_0$ is not ranked first in $\sigma^*$. Thus $\sigma^*(c_0)=r \neq 1$. There exists $c_1 \in \C$ such that $\sigma^*(c_1)=r-1$. Let $\tau \in \S_\dim$ be such that $\tau=\sigma^*$ except that $\tau(c_0)=r-1$ and $\tau(c_1)=r$.

Observe that for all $s \in \S_\dim$, $$d_K(\sigma^*,s)=d_K(\tau,s)+\delta_s(c_0,c_1),$$ where $\delta_s(c_0,c_1)=1$ if $s(c_0)<s(c_1)$ and $\delta_s(c_0,c_1)=-1$ if $s(c_0)>s(c_1)$. Therefore, $$\sum_{v=1}^n d_K(\sigma^*,\sigma_v)=\sum_{v=1}^n d_K(\tau,\sigma_v)+ \sum_{v=1}^n \delta_{\sigma_v}(c_0,c_1).$$
As $c_0$ is a Condorcet winner, we have more elements in the set $\{v \in \{1,\dots n\},\; \sigma_v(c_0)<\sigma_v(c_1)\}$ than in the set  $ \{v \in \{1,\dots n\},\; \sigma_v(c_0)<\sigma_v(c_1)\}$. This implies that $ \sum_{v=1}^n \delta_{\sigma_v}(c_0,c_1)>0$  and thus that $\sum_{v=1}^n  d_K(\sigma^*,\sigma_v)>\sum_{v=1}^n d_K(\tau,\sigma_v)$. Hence, $\sigma^*$ cannot be a median. Therefore, by absurd, it means that $\sigma^*(c_0)=1$.

So the consensus Kendall-based voting rule satisfies the Condorcet-winner property. Note that it is well-known that Kemeny voting satisfies Condorcet-winner property.

\item \textbf{Kendall-based $p$-Fréchet  for $p>1$,  Cayley, Hamming -based $p$-Fréchet  for $p\geq1$ and $q$-Minkowski-Hölder -based $p$-Fréchet  for $p=1$}
\begin{itemize}
    \item \textit{Kendall-based $p$-Fréchet  for $p>1$} \newline
Consider a setting with 5 voters and 3 candidates, with the voting preferences given are according to \cref{tab:condorcetKendall}. 

\begin{minipage}{\linewidth}
    \centering
    \begin{tabular}{cccccc} \toprule
            & $v_1$ & $v_2$ & $v_3$ & $v_4$ & $v_5$   \\ \midrule
        $c_1$ &   1 & 1 & 1 & 3 & 2\\
        $c_2$ &  2 & 2 & 2 & 2 & 1\\
        $c_3$ &  3 & 3 & 3 & 1 & 3\\  \bottomrule
    \end{tabular}
    \captionof{table}{Counter-example for the proof of point 2. (Kendall) of \cref{prop:condowinner}.}
    \label{tab:condorcetKendall}
\end{minipage}
 
 $c_1$ is obviously the Condorcet winner of the election. Then, for $p \geq 1$, 
\begin{align*}
\displaystyle \sum_{v=1}^5 d_K^p((c_1,c_2,c_3), v_v) & = 1+3^p\\
\sum_{v=1}^5 d_K^p((c_1,c_3,c_2), v_v) & = 3 + 2 \times 2^p\\
\sum_{v=1}^5 d_K^p((c_2,c_1,c_3), v_v) & = 3 + 2^p\\
\sum_{v=1}^5 d_K^p((c_2,c_3,c_1), v_v) & = 1 + 4 \times 2^p\\
\sum_{v=1}^5 d_K^p((c_3,c_1,c_2), v_v) & = 2 + 3 \times 2^p\\
\sum_{v=1}^5 d_K^p((c_3,c_2,c_1), v_v) & = 2^p + 3\times 3^p\\
\end{align*}

If $p=1$, $c_1$ is the winner of Kendall consensus ranking.
 
 But  $\forall p\geq 2$, $\sum_{v=1}^5 d_K((c_1,c_2,c_3), v_v)^p > \sum_{v=1}^5 d_K((c_2,c_1,c_3), v_v)^p$ and then $c_2$ is the winner of the election with the voting rule (Kendall,$p$) for $p \geq 2$. Other calculations are left to the reader.

\end{itemize}

\item For the following with Cayley and Hamming, consider a setting with 7 voters and 3 candidates, with the voting preferences given are according to \cref{tab:condorcet}.

\begin{minipage}{\linewidth}
    \centering
    \begin{tabular}{cccccccc} \toprule
            & $v_1$ & $v_2$ & $v_3$ & $v_4$ & $v_5$ & $v_6$ & $v_7$  \\ \midrule
        $c_1$ &   1 & 1 & 2 & 2 & 2 & 2 & 3\\
        $c_2$ &  2 & 2 & 1 & 1 & 3 & 3 & 1\\
        $c_3$ &  3 & 3 & 3 & 3 & 1 & 1 & 2\\  \bottomrule
    \end{tabular}
    \captionof{table}{Counter-example for the proof of point 2 of \cref{prop:condowinner}.}
    \label{tab:condorcet}
\end{minipage}

In this setting, the Condorcet winner is candidate $c_1$ as 4 voters prefer $c_1$ to $c_2$ and 4 voters prefer $c_1$ to $c_3$.

\begin{itemize}
\item \textit{Hamming-based $p$-Fréchet  for $p\geq1$}
With the voting preferences given are according to \cref{tab:condorcet},
 \begin{align*}
\displaystyle \sum_{j=1}^7 d_H^p((1,2,3)^T,v_j) & =  2^p + 2^p + 3^p + 3^p + 3^p\\
\displaystyle \sum_{j=1}^7 d_H^p((2,1,3)^T,v_j) & =  2^p + 2^p + 2^p + 2^p + 2^p\\
\displaystyle \sum_{j=1}^7 d_H^p((2,3,1)^T,v_j) & =  3^p + 3^p + 2^p + 2^p + 3^p\\
\displaystyle \sum_{j=1}^7 d_H^p((3,1,2)^T,v_j) & =  3^p + 3^p + 2^p + 2^p + 3^p + 3^p\\
\end{align*}
Other calculations are left to the reader and we note that $(2,1,3)^T$ minimizes the Hamming-based $p$-Fréchet mean for $p\geq1$ and then $c_2$ is the winner of the election. \\

 \item \textit{Cayley-based $p$-Fréchet  for $p\geq1$}
 
With the voting preferences given are according to \cref{tab:condorcet}, 
  \begin{align*}
\displaystyle \sum_{j=1}^7 d_C^p((1,2,3)^T,v_j) & =  1   + 1   + 2^p + 2^p + 2^p\\
\displaystyle \sum_{j=1}^7 d_C^p((2,1,3)^T,v_j) & =  1   + 1   + 1   + 1   + 1\\
\displaystyle \sum_{j=1}^7 d_C^p((2,3,1)^T,v_j) & =  2^p + 2^p + 1   + 1   + 2^p\\
\displaystyle \sum_{j=1}^7 d_C^p((3,1,2)^T,v_j) & =  2^p + 2^p + 1   + 1   + 2^p + 2^p\\
\end{align*}
Other calculations are left to the reader and we note that $(2,1,3)^T$ minimizes the Cayley-based $p$-Fréchet mean for $p\geq1$ and then $c_2$ is the winner of the election.
 \end{itemize}

 \item For \textit{$q$-Minkowski-Hölder based $p$-Fréchet  for $p=1$}, let's consider the previous counter example.
 
 One can show that, for $q \geq 1$, and the preferences expressed in table \ref{tab:condorcet}:
\begin{align}
\displaystyle \sum_{j=1}^7 d_q((1,2,3)^T,v_j) & =  2 \times 2^{1/q}  + 3 \times (2^q+2)^{1/q} +  0 \times 2^{(q+1)/q}\\
\displaystyle \sum_{j=1}^7 d_q((2,1,3)^T,v_j) & =  3 \times 2^{1/q}  + 0 \times (2^q+2)^{1/q} +  2 \times 2^{(q+1)/q}\\
\displaystyle \sum_{j=1}^7 d_q((1,3,2)^T,v_j) & =  4 \times 2^{1/q}  + 2 \times (2^q+2)^{1/q} +  1 \times 2^{(q+1)/q}\\
\displaystyle \sum_{j=1}^7 d_q((3,1,2)^T,v_j) & =  3 \times 2^{1/q}  + 4 \times (2^q+2)^{1/q} +  0 \times 2^{(q+1)/q}\\
\displaystyle \sum_{j=1}^7 d_q((2,3,1)^T,v_j) & =  0 \times 2^{1/q}  + 3 \times (2^q+2)^{1/q} +  2 \times 2^{(q+1)/q}\\
\displaystyle \sum_{j=1}^7 d_q((3,2,1)^T,v_j) & =  3 \times 2^{1/q}  + 2 \times (2^q+2)^{1/q} +  2 \times 2^{(q+1)/q}\\
\end{align}
This sum is minimum for $(2,1,3)^T$ for all $q \geq 1$. Basic considerations as $2^{1/q} \leq (2^q+2)^{1/q} \leq 2^{(q+1)/q}$ mean to the act that only $(1,2,3)^T$ and $(2,1,3)^T$ can minimises the sum. Finally, $(2,1,3)^T$ is the permutation minimizing the sum.

\end{enumerate}

\subsection*{\bf Proof of \cref{prop:monotonie}}

\begin{enumerate}
   \item Consider $\Phi$ a voting situation with a winner $c^*$ with a given voting procedure. Suppose that there exists a voter $v_0$ such that $\Phi(c^*,v_0)= \alpha \neq 1$. Let $c_0$ such that $\Phi(c_0,v_0)=\alpha-1$. 
Consider another voting situation $\widetilde{\Phi}=\Phi$ except that $\widetilde{\Phi}(c^*,v_0)=\alpha-1$ and $\widetilde{\Phi}(c_0,v_0)=\alpha$. The voting procedure satisfies Monotonicity if the winner for $\tilde{\Phi}$ is still $c^*$.

Denote $\Pi_n$ and $\tilde{\Pi}_n$ the empirical distributions associated to the columns of $\Phi$ and $\tilde{\Phi}$ respectively. Without loss of generality, we suppose that the voter for whom we swap the  ranks of $c^*$ and $c_0$ is the voter $v_0=n$. We denote $\sigma_n$ and $\tilde{\sigma}_n$ the associated rankings. Consider $U_{d, \Pi,p}(\cdot)$ the $p$-Fréchet functional associated to the initial votes, and  $U_{d, \tilde{\Pi}_n, p}(\cdot)$ the functional obtained after the swap. Let $\sigma^*$ be a Fréchet mean associated to the distance $d()$ in the initial voting setting. We would like to show that $\sigma^*$ is still a Fréchet mean in the new setting, namely that it minimizes $U_{d, \tilde{\Pi}_n, p}$. 

Let $\sigma\in \S_\dim$. We have 
\begin{align}\label{ineq: mon}
    n U_{d, \tilde{\Pi}_n, p}(\sigma^*)-n U_{d, \tilde{\Pi}_n, p}(\sigma)&=n U_{d, \Pi_n, p}(\sigma^*)-n U_{d,\Pi_n,p}(\sigma)\\    
    &\quad +\bigl( d^p(\sigma^*, \tilde{\sigma}_n)-d^p(\sigma^*,\sigma_n)\bigr)-\bigl( d^p(\sigma,\tilde{\sigma}_n)-d^p(\sigma,\sigma_n)\bigr).\nonumber
\end{align}
We want to prove that $n U_{d, \tilde{\Pi}_n, p}(\sigma^*)-n U_{d, \tilde{\Pi}_n, p}(\sigma)<0$.
We now distinguish with respect to the distance used and the value of $q$.

    Let $p=1$. Suppose that $\sigma^*$ is a Fréchet mean associated to $U_{d_1, \Pi_n, 1}(\cdot)$. Then $U_{d_1, \Pi,1}(\sigma^*)<U_{d_1, \Pi_n, 1}(\sigma)$ for all $\sigma\in \S_\dim\setminus\S^*_{1,d_1,\Pi_n}$.

    First observe that, since the $1$-Minkowski-Hölder distance between two permutations is always even, this implies in particular that  $n U_{d_1, \Pi_n, 1}(\sigma^*)\le n U_{d_1, \Pi_n, 1}(\sigma)-2$.
    
    Secondly, as there is only one swap between ${\sigma}_n$ and $\tilde{\sigma}_n$, with a difference one between the values, we have, for any permutation $\sigma$, 
    $$|d_1(\sigma,\tilde{\sigma}_n)-d_1(\sigma,\sigma_n)|\le 2.$$  
    Replacing this in \eqref{ineq: mon}, we obtain
    $$ n U_{d_1, \tilde{\Pi}_n, 1}(\sigma^*)-n U_{d_1, \tilde{\Pi}_n, 1}(\sigma)\le 
d_1(\sigma^*,\tilde{\sigma}_n)-d_1(\sigma^*,\sigma_n)$$
Lastly, observe that $d(\sigma^{\star},\tilde{\sigma}_n)-d(\sigma^{\star},\sigma_n)$ depends only on the rankings of $c^*$ and $c_0$. Indeed, $d(\sigma^{\star},\tilde{\sigma}_n)-d(\sigma^{\star},\sigma_n)=-2$ if $\sigma^{\star}(c_0)\ge \alpha$ and $0$ otherwise. Thus $U_{d_1, \tilde{\Pi},1}(\sigma^{\star})\le U_{d_1, \tilde{\Pi},1}(\sigma)$ for all $\sigma \in \mathfrak{S}_m\setminus\mathfrak{S}^*_{1,d_1,\Pi}$ and all permutations in $\mathfrak{S}^*_{1,d_1,\Pi}$ are also a Fréchet mean after the vote swap. If, in addition $\sigma^{\star}(c_0)\ge \alpha$ then the inequality is strict implying that $\mathfrak{S}^*_{1,d_1,{\Pi}}\subseteq \mathfrak{S}^*_{1,d_1,\tilde{\Pi}}$ and thus that $c^*$ remains  the unique winner of the election.

\item Consider the following voting preferences: $a=(A,B,C,D,E)$, $b=(D,C,A,B,E)$, $c=(E,D,C,A,B)$. Suppose that we have $n=5$ voters and they vote as follows: $\sigma_1=\sigma_2=a$, $\sigma_3=\sigma_4=b$ and $\sigma_5=c$. Then one can check that the Fréchet mean associated to the consensus ranking for the Hamming distance is $a$ and thus the winner of the election is $A$. Now if the fifth voter changes his rankings, swapping the place of $A$ and $C$, $\tilde{\sigma}_5$ becomes $(E,D,A,C,B)$ and $b$ becomes a Fréchet mean for $\tilde{\Pi}$, leading to a new election winner, $D$.
\end{enumerate}

\subsection*{\bf Proof of \cref{prop:IIA}}

Independence to Losers property:
\begin{enumerate}
    \item Kendall-based consensus deepest voting is Kemeny voting rule, which is known to not satisfy Independence to Losers as all Condorcet methods (see \cite{young1978consistent}).
    
    \item  We want to prove that all deepest voting rules based on $q$-Minkowski-Hölder distance, for all $q\geq 1$, do not satisfy Independence to Losers. 

As a counterexample, let suppose the situation with 2 candidates and 7 voters, with rankings given by \cref{tab:iia_minkowski_1}.

\begin{minipage}{\linewidth}
\centering
\begin{tabular}{lccccccc} \toprule
& $v_1$ & $v_2$ & $v_3$ & $v_4$ & $v_5 $& $v_6$ & $v_7$ \\ \midrule
 $c_1 $  & 1 & 1 & 1 & 1 & 2 & 2 & 2 \\
$c_2 $ & 2 & 2 & 2 & 2 & 1 & 1 & 1 \\ \bottomrule
\end{tabular}
\captionof{table}{Counter-example for the proof that voting rules based on a Minkowski-Hölder distance do not satisfy Independence to Losers. Case 1, with 2 candidates.}
\label{tab:iia_minkowski_1}
\end{minipage}

It is obvious that for all $q \geq 1$, the deepest permutation is $(1,2)^\top$, and therefore candidate $c_1$ is the winner of the deepest voting with $q$-Minkowski-Hölder distance.

Let introduce a candidate $c_3$ such that the preferences are now the ones in \cref{tab:iia_minkowski_2}. Note that the pairwise comparisons between $c_1$ and $c_2$ do not change.

\begin{minipage}{\linewidth}
\centering
\begin{tabular}{lccccccc} \hline
& $v_1$ & $v_2$ & $v_3$ & $v_4$ & $v_5 $& $v_6$ & $v_7$ \\ \hline
 $c_1 $  & 1 & 1 & 2 & 2 & 3 & 3 & 3 \\
$c_2 $ & 2 & 2 & 3 & 3 & 1 & 1 & 1 \\ 
$c_3 $ & 3 & 3 & 1 & 1 & 2 & 2 & 2 \\ \hline
\end{tabular}
\captionof{table}{Counter-example for the proof that voting rules based on a Minkowski-Hölder distance do not satisfy Independence to Losers. Case 2, with 3 candidates. The voting situation for the first 2 candidates is similar to Case 1 in \cref{tab:iia_minkowski_1}.}
\label{tab:iia_minkowski_2}
\end{minipage}

It is easy to check that for all $q\geq 1$ (calculus are left to the reader) the deepest permutation is $(3,1,2)^\top$ and therefore $c_2$ is the deepest winner of the election. Introducing $c_3$ changes the winner from $c_1$ to $c_2$, which proves that Independence to Losers is not satisfied.

\item To prove that Hamming-based deepest voting do not satisfy Independence to Losers, the same counterexample of the proof of Minkowski-Holder-based deepest voting can be considered. 

\item The same counterexample can be used for the proof with Cayley-based deepest voting, even if the winner is not unique in this case.

\end{enumerate}

\end{document}